\documentclass[leqno]{article}
\usepackage{ltexpprt}
\usepackage{graphicx}
\usepackage{caption}
\usepackage{subcaption}
\usepackage{ amssymb, amsmath}
\newtheorem{claim}{Claim}
\newcommand*{\QEDA}{\hfill\ensuremath{\blacksquare}}
\newcommand{\hide}[1]{}
\title{\bf Analyzing the Optimal Neighborhood: Algorithms for Budgeted and Partial Connected Dominating Set Problems\thanks{The work is supported by NSF grants CCF 1217890 and CCF 0937865.}}

\author{Samir Khuller \thanks{ Computer Science Department, University of Maryland, College Park. E-mail: \texttt{samir@cs.umd.edu}}
        \and Manish Purohit \thanks{ Computer Science Department, University of Maryland, College Park. E-mail: \texttt{manishp@cs.umd.edu}}
	\and Kanthi K. Sarpatwar \thanks{ Computer Science Department, University of Maryland, College Park. E-mail: \texttt{kasarpa@cs.umd.edu}}}
\date{}

\DeclareMathOperator*{\argmax}{arg\,max}

\usepackage{mathtools}
\usepackage{algpseudocode}
\usepackage{algorithmicx}
\begin{document}
\maketitle

\begin{abstract}
\small\baselineskip=9pt  
We study  \emph{partial} and \emph{budgeted} versions of the well studied 
connected dominating set problem. 
In the partial connected dominating set problem ({\sc Pcds}), we are given an 
undirected graph $G = (V,E)$ and an integer $n'$, and the goal is to find 
a minimum subset of vertices that induces a connected subgraph of $G$
and dominates at least $n'$ vertices. 
We obtain the first polynomial time algorithm with an 
$O(\ln \Delta)$ approximation factor for this problem, 
thereby significantly extending the results of Guha and Khuller (Algorithmica, Vol. 20(4), Pages 374-387, 1998) for 
the connected dominating set problem. 
We note that none of the methods developed earlier 
can be applied directly to solve this problem. 
In the budgeted connected dominating set problem ({\sc Bcds}), there is a budget 
on the number of vertices we can select, and the goal is to dominate as many 
vertices as possible. 
We obtain a $\frac{1}{13}(1-\frac{1}{e})$ approximation algorithm for this 
problem. Finally, we show that our techniques extend to a more general 
setting where the profit function associated with a subset of vertices is a 
\lq\lq{}special\rq\rq{} submodular function. 
This generalization captures the connected dominating set problem with capacities and/or weighted profits as special cases.
This implies a $O(\ln q)$ approximation (where $q$ denotes the quota) and an $O(1)$ 
approximation algorithms for the partial and budgeted versions of these 
problems.  While the algorithms are simple, the results make a surprising
use of the greedy set cover framework in defining a useful profit function.
\end{abstract}
%\titlepage{}

\section{Introduction}

A \emph{connected dominating set} ({\sc Cds}) in a graph is a dominating set 
that induces a connected subgraph.
The {\sc Cds} problem, which seeks to find the minimum such set, has been widely
studied~\cite{alzoubi, aravind,GuhaKhuller,wu,das1997routing,Srinivasan, du2013connected,liu:liang,cheng2005,cheng2007} starting from the 
work of Guha and Khuller~\cite{GuhaKhuller}.
The {\sc Cds} problem is NP-hard and thus the literature has focused
on the development of fast polynomial time approximation algorithms. 
For general graphs, Guha and Khuller \cite{GuhaKhuller} propose an algorithm 
with a $\ln \Delta+3$ approximation factor, where $\Delta$ is the maximum degree of any vertex. 
Better approximation algorithms are known in special classes of graphs. 
For the case of planar~\cite{Hajiaghayi} or geometric unit 
disk graphs~\cite{du} polynomial time approximation schemes (PTAS) are known.  
This problem has also been extensively studied in the distributed 
setting~\cite{alzoubi, aravind}.
Not surprisingly, {\sc Cds}  problem in general graphs is at least as hard to approximate
as the set cover problem for which a hardness result of $(1-\epsilon) \log n$ (unless $NP\subseteq DTIME(n^{O(\log\log{n})})$) follows
by the work of Feige \cite{Feige}. % essentially proving that the bound
%obtained by the greedy set cover algorithm is tight.

%==============================================

{\sc Cds} has become an extremely popular topic, 
for example, the recent book by Du and Wan \cite{du2013connected}
focuses on 
the study of ad hoc wireless networks as {\sc Cds}s provide a platform for routing on
such networks. In these ad hoc wireless networks,  a {\sc Cds} can act as a virtual backbone so that only nodes belonging to the {\sc Cds} are responsible for packet forwarding and routing. Minimizing the number of nodes in the virtual backbone leads to increased network lifetime, and lesser bandwidth usage due to control packets, and hence the {\sc Cds} problem has been extensively studied and applied to create such virtual backbones. 

One shortcoming of using a {\sc Cds} as a virtual backbone is that a few distant clients (outliers) can have the undesirable effect of increasing the size of the {\sc Cds} without improving the quality of service to a majority of the clients. In such scenarios, it is often desirable to obtain a much smaller backbone that provides services to, say,  (at least) 90\% of the clients. Liu and Liang \cite{liu:liang} study this problem of finding a minimum \emph{partial connected dominating set} in wireless sensor networks (geometric disk graphs) and provide heuristics 
(without guarantees) for the same. 

A complementary problem is the budgeted {\sc Cds} problem where we have a budget of $k$ nodes, and we wish to find a connected subset of $k$ nodes which dominate as many vertices as possible. Budgeted domination has been studied in sensor networks where bandwidth constraints limit the number of sensors we can choose and the objective is to maximize the number of targets covered \cite{cheng2005,cheng2007}. 

Another application arises in the context of social networks. Consider a social network where vertices of the network correspond to people and an edge joins two vertices if the corresponding people influence each other. Avrachenkov et al.\@ \cite{Konstantin} consider the problem of choosing $k$ connected vertices having maximum total influence in a social network using local information only (i.e., the neighborhood of a vertex is revealed only after the vertex is “bought”) and provide heuristics (without guarantees) for the same. 
Borgs et al.~\cite{borgs2012power} show that no local algorithm for the partial dominating set problem can provide an approximation guarantee better than $O(\sqrt{n})$. As the influence of a set of vertices is simply the number of dominated vertices, these problems are exactly the budgeted and partial connected dominating set problems with the additional restriction of local only information.

Budgeted versions of set cover (known as max-coverage)\footnote{Here
instead of finding the smallest sub-collection of sets to cover
a given set of elements, we fix a budget on number of
sets we wish to pick with the objective of maximizing the number of
covered elements.} are well understood and the standard greedy algorithm is known to give the optimal $1 - \frac{1}{e}$ approximation~\cite{NWF}. Khuller et al.\@~\cite{Khuller2} give a $(1 - \frac{1}{e})$ approximation algorithm for a generalized version with costs on sets. In addition, we may consider a partial version of the set cover problem, also known as partial covering, in which we wish to pick the minimum number of sets to cover a pre-specified number of elements. Kearns~\cite{kearns} first showed that greedy gives a $2H_n+3$ approximation guarantee (where $n$ is the ground set cardinality and $H_n$ is the $n^{th}$ harmonic number), which was improved by Slav{\'\i}k~\cite{slavik} to obtain a guarantee of {\it min}$(H_{n'}, H_{\Delta})$(where $n'$ is the minimum coverage required and $\Delta$ is the maximum size of any set).  Wolsey~\cite{wolsey82} considered the more general \emph{submodular cover} problem and showed that the simple greedy delivers a best possible $\ln n$ approximation.   
 For the case where each element belongs to at most $f$ (called the \emph{frequency}) different sets, Gandhi et al.\@ \cite{GKS}, using a primal-dual algorithm, and Bar-Yehuda~\cite{bar-yehuda}, using the local-ratio technique, achieve an $f$-approximation guarantee.

Unfortunately for both the budgeted and partial
versions of the {\sc Cds} problem, greedy approaches based on prior
methods fail.
The fundamental reason is that while the greedy algorithm
works well as a method for rapidly \lq\lq{}covering\rq\rq{} nodes, the cost
to connect different chosen nodes can be extremely high if the
chosen nodes are far apart. On the other hand if we try to maintain
a connected subset, then we cannot necessarily select nodes
from dense regions of the graph.
In fact, none of the approaches in
the work by Guha and Khuller~\cite{GuhaKhuller} appear  to 
extend to these versions. 

Partial and budgeted optimization problems have been extensively studied in 
the literature.
Most of these problems, with the exception of partial and 
budgeted set cover, required significantly different techniques and ideas 
from the corresponding \lq\lq{}complete\rq\rq{} versions. 
We will now cite several such problems.

The best example is the minimum spanning tree problem, which is well known 
to be polynomial time solvable. However 
the partial version of this problem where 
we look for a minimum cost tree which spans at least $k$ vertices is NP-hard~\cite{ravi}. 
A series of approximation techniques~\cite{arora, baruch,blum,garg1} finally 
resulted in a $2$-approximation~\cite{garg2} for the problem. 

Partial versions of several classic location problems like $k$-center and 
$k$-median have required new techniques as well. The partial 
$k$-center problem, which is also called the outlier $k$-center problem or the
robust $k$-center problem, requires us to minimize the maximum distance to the 
\lq\lq{}best\rq\rq{} $n'$ nodes (while the complete version requires us to 
consider all the nodes) to the centers. 
Charikar et al.\@~\cite{khuller:outlier} gave a 3-approximation algorithm whose 
analysis was significantly different from the classic $k$-center 2-approximation
algorithm~\cite{Gonzalez,hochbaum1986unified}. 
Chen~\cite{chen2008constant} gives a constant approximation for outlier 
$k$-median problem, while Charikar et al.\@~\cite{khuller:outlier} gave a 4-approximation for the outlier uncapacitated facility location problem.

Several other optimization problems need special techniques 
to tackle the corresponding partial versions. Notable examples of such 
problems, include \emph{partial vertex cover}~\cite{srinivasan2001, bar:cappvc,mestre:pvc,GKS,bshouty,halperin2002},
quota Steiner tree problems~\cite{johnson2000prize}, budgeted and partial 
node weighted Steiner tree problems~\cite{Moss,vahid}, and scheduling with 
outliers~\cite{CharikarKhuller,Gupta:sched}. We end this 
subsection by noting that partial versions of some optimization problems are 
completely inapproximable even though, the corresponding complete version has a
small constant approximation algorithm. The best example of this is the 
\emph{robust subset resource replication} problem studied by 
Khuller et al.\@~\cite{khu:sah:kan}.

\subsection{Other Related Work}

In the \emph{group Steiner tree problem}, we are given a graph $G=(V,E)$, an associated cost function $c:E\rightarrow \mathbb{R}^+\cup \{0\}$, and a collection of groups of vertices $g_1, g_2, \ldots, g_k$. The goal is to find a minimum cost tree that contains 
at least one vertex from each group. It can be observed that the connected dominating set problem reduces to the group Steiner tree problem by creating a group for every vertex containing the neighborhood of that vertex. Garg et al.\@~\cite{Garg} obtain a $O(\log (\max_{i\in[k]} |g_i|)\log k)$ approximation for this problem, in the special case when the graph is a tree. Using a decomposition result due to Bartal~\cite{Bartal1}, Garg et al.\@~\cite{Garg} extend the tree result to obtain a $O(\log^3n\log k)$ approximation algorithm for arbitrary graphs.  Fakcharoenphol et al.~\cite{FRT} improve  Bartal's decomposition result, consequently obtaining a $O(\log^2n\log k)$ approximation for the group Steiner tree problem in arbitrary graphs. Halperin et al.\@~\cite{Krauthgamer} note that Garg et al.\@~\cite{Garg}'s algorithm also gives a $O(\log (\max_{i\in[k]}|g_i|))$ approximation for the budgeted group Steiner tree problem on trees. They also show a $\log^{2-\epsilon}k$ hardness of approximation for the (partial) group Steiner problem and a $\log^{1-\epsilon}k$ hardness of approximation for the budgeted version.
Chekuri et al.~\cite{CEK} gave a combinatorial algorithm for the group Steiner tree problem on trees, with an approximation guarantee of $O((\log\sum_i |g_i|)^{1+\epsilon}\log k)$. Calinescu and Zelikovsky~\cite{CZ} extended Chekuri et al.~\cite{CEK}'s result to the more general problem of \emph{polymatroid Steiner tree}. 
 %Hence, the polylogarithmic approximation ratios by Garg et al.\@~\cite{Garg} are applicable to the budgeted and partial connecting dominating set problems.

\subsection{Our Contributions}
Our results can be summarized as follows
\begin{itemize}
\item[-]  In Section~\ref{sec:partial}, we obtain the first $O(\ln \Delta)$ approximation algorithm for the {\sc Pcds} problem. To be precise, our  approximation guarantee is $4\ln \Delta + 2 + o(1)$, where $\Delta$ is the maximum degree. 
\item[-]  In Section~\ref{sec:budgeted}, we obtain a $\frac{1}{13}(1-\frac{1}{e})$-approximation algorithm for the {\sc Bcds} problem. This is the first constant approximation known for {\sc Bcds}.
\item[-]  In Section~\ref{sec:bgcd}, we generalize the above problems to a special kind of submodular optimization problem (to be defined later), which has \emph{capacitated connected dominating set} problem and \emph{weighted profit connected dominating set} problem as special cases. Again, we obtain $O(\ln q)$ and $\frac{1}{13}(1-\frac{1}{e})$ approximation algorithms for the partial and budgeted version of this problem respectively where $q$ denotes the quota for the partial version. 
\end{itemize}

%Since the result for the partial {\sc Cds} problem is simpler, we first
%focus on explaining that algorithm.

\section{Preliminaries}
\label{sec:preliminaries}
We now proceed to formally define the problems we consider in this paper.

\noindent {\sc Partial Connected Dominating Set Problem ({\sc Pcds})}.
Given an undirected graph $G=(V,E)$, and an integer (quota) $n'$, find a
minimum size subset $S \subseteq V$ of vertices such that the graph induced 
by $S$ is connected, and $S$ dominates at least $n'$ vertices.

\noindent {\sc Budgeted Connected Dominating Set Problem ({\sc Bcds})}.
Given an undirected graph $G=(V,E)$, and an integer (budget) $k$, find a 
subset $S \subseteq V$ of at most $k$ vertices such that the graph induced by 
$S$ is connected, and the number of vertices dominated by $S$ is maximized.

Before defining the remaining problems, we introduce the notion of \emph{special submodularity}.

\noindent {\sc Special Submodular Function}.
Let $G=(V,E)$ be an arbitrary graph.  A function $f:2^V\rightarrow \mathbb{Z}^+\cup \{0\}$, is said to have the \emph{special submodular} property if it satisfies the following- %conditions.
\begin{itemize}
\item $f$ is submodular. That is $f(A \cup \{v\}) - f(A) \geq f(B\cup \{v\}) - f(B) \quad\forall A,B,v$ such that $A\subseteq B\subseteq V$.
\item $f_A(X) = f_{A \cup B}(X)$, if $N(X) \cap N(B) = \phi$ $\quad\forall X, A, B \subseteq V$.%\footnote{Here $N(\cdot)$ denotes the neighborhood} 
\end{itemize}
where $f_A(X) = f(A \cup X) - f(A)$ is the marginal profit of $X$ given $A$ and $N(X)$ denotes the neighborhood of $X$, including $X$ itself.

 We now define the generalized versions of {\sc Pcds} and {\sc Bcds}. 

\noindent {\sc Partial Generalized Connected Dominating Set Problem (Pgcds)}.
Given a graph $G = (V,E)$, an integer (quota)  $q$, and a monotone \emph{special submodular} profit function $f : 2^V \rightarrow \mathbb{Z^+}\cup \{0\}$, find a subset $S \subseteq V$ of minimum size, such that $f(S) \geq q$ and $S$ induces a connected subgraph in $G$. 

\noindent {\sc Budgeted Generalized Connected Dominating Set Problem (Bgcds)}.
Given a graph $G = (V,E)$, a budget $k$, and a monotone \emph{special submodular} profit function $f : 2^V \rightarrow \mathbb{Z^+}\cup \{0\}$, find a subset $S \subseteq V$ which maximizes $f(S)$ such that $|S| \leq k$ and $S$ induces a connected subgraph of $G$.

 These problems capture the following variants of partial and budgeted connected dominating set problems.
\begin{enumerate}
\item {\sc Weighted Profit Connected Dominating Set.} In this variant, each vertex has an arbitrary profit which is obtained if it is dominated by some chosen vertex.
\item {\sc Capacitated Connected Dominating Set.} In this variant, each vertex has a capacity which is the number of vertices it can dominate. 
\end{enumerate}

 For all of our algorithms we will be using an algorithm for the Quota Steiner Tree problem ({\sc Qst}) as a subroutine. We now define the {\sc Qst} problem and mention relevant results.

\noindent {\sc Quota Steiner Tree Problem ({\sc Qst})}. Given an undirected graph $G=(V,E)$, a profit function $p:V \rightarrow Z^+ \cup \{0\}$ on the vertices, a cost function $c:E \rightarrow Z^+ \cup \{0\}$ on the edges, and an integer (quota) $q$, find a subtree $T$ that minimizes $\sum_{e \in E(T)} c(e)$, subject to $\sum_{v \in V(T)} p(v) \geq q$.

Johnson et al.\@~\cite{johnson2000prize} studied the {\sc Qst} problem and showed that an $\alpha$-approximation algorithm for the $k$-MST problem can be adapted to obtain an $\alpha$-approximation algorithm for the Quota Steiner Tree problem. Using this result along with the 2-approximation for $k$-MST by Garg~\cite{garg2}, gives us the following theorem. 

\begin{theorem}[\cite{johnson2000prize,garg2}]
\label{thm:QST}
There is a 2-approximation algorithm for the Quota Steiner Tree Problem.
\end{theorem}

% === NEWLY ADDED PART ===
\section{Shortcomings of Prior Approaches.}

We now describe the three approaches taken by Guha and Khuller~\cite{GuhaKhuller} to solve the {\sc Cds} problem and show why none of these approaches extend directly for the budgeted and partial coverage variants. 

\emph{Algorithm 1.}
The first algorithm is a ``one step look-ahead'' greedy algorithm where they iteratively grow a tree by selecting a \emph{pair} of vertices that together cover the most number of previously uncovered vertices.  Figure~\ref{fig:bad_example} shows a bad instance on which a $c$-step look-ahead greedy algorithm fails for the {\sc Bcds} and {\sc Pcds}. The instance contains $k$ ``spiders'' whose heads (vertices with degree $> 2$) are connected by paths of length $c+1$. The spider heads are the only vertices that offer profit greater than 3. We show that on this graph, there are {\sc Bcds} and {\sc Pcds} instances that can perform very poorly. 
Consider a {\sc Bcds} instance on the graph, with a budget $k+(c+1)(k-1)$. Clearly the optimal solution picks the path connecting all the spider heads, so that the total coverage is $(M+1)k + (c+1)(k-1)$. On the other hand, the $c$-step look-ahead greedy algorithm, might get stuck inside one of the spiders and may end up selecting as many as $M+1$ vertices from it. This is because, despite the look-ahead capability of the algorithm, the spider legs will become indistinguishable from the optimal path. For a sufficiently large value of $M$, the $c$-step look-ahead algorithm might use up all its budget on a single spider, there by obtaining a coverage of $O(M+k)$. Thus in the worst case the look-ahead greedy algorithm could have a $\Omega(k)$ approximation guarantee. Using a similar argument, we can show that, for the {\sc Pcds} instance on the graph with quota $Mk$, the approximation guarantee could be $\Omega(M)$.

\begin{figure}[hnbt]

\centering
\includegraphics[width=0.45\textwidth]{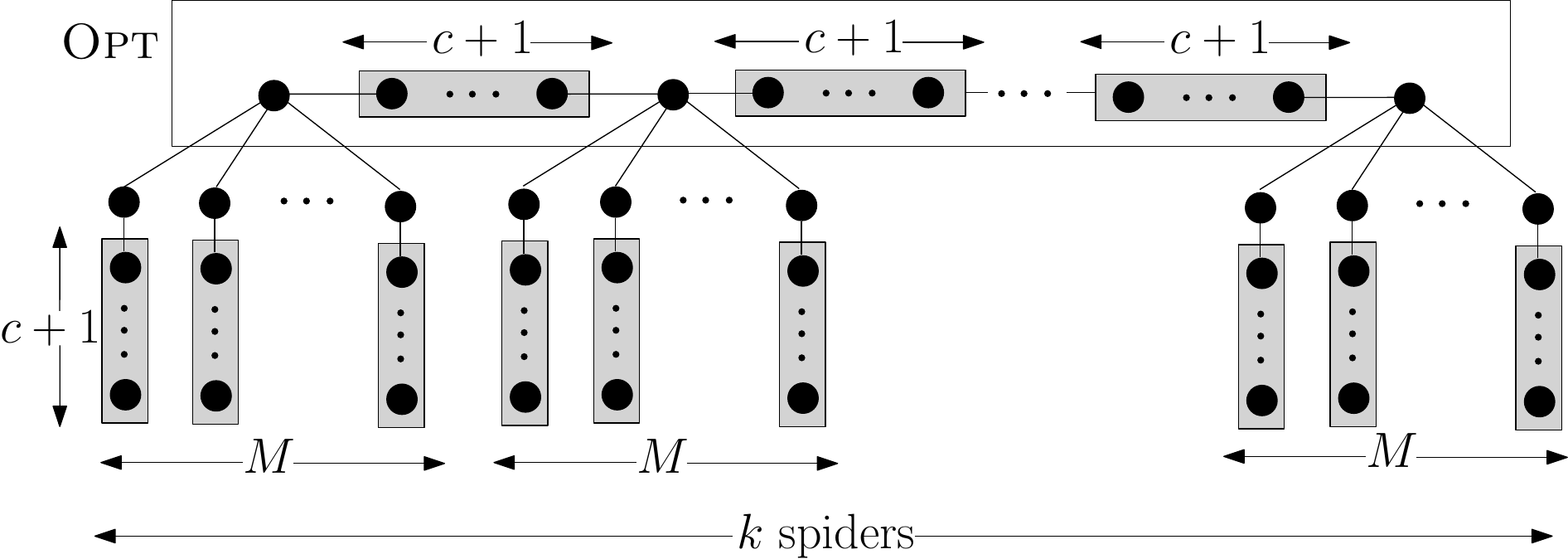}

\caption{A bad example for the $c$-step look-ahead greedy algorithm}
\label{fig:bad_example}

\end{figure}

\emph{Algorithm 2.}
The second algorithm is to find a dominating set $D$ and run a Steiner tree algorithm with the vertices in $D$ as terminals. Since the optimal connected dominating set, by definition, is a tree that dominates $D$, we can show that there exists a Steiner tree of low cost with the set $D$ as terminals. Using a constant factor approximation algorithm for the Steiner tree problem, we obtain a $O(\ln n)$ approximation for the connected dominating set. However, for the partial and budgeted versions, the optimal solution does not dominate all vertices and hence it's not possible to bound the cost of the Steiner tree in terms of the optimal solution.

\emph{Algorithm 3.}
The final algorithm builds unconnected components greedily and owing to the fact that every vertex has to be dominated, makes sure that the constructed components be connected cheaply. Again this approach fails in the partial and budgeted case because the components created when we have dominated a specified number of vertices could be far apart. 

\section{Partial Connected Dominating Set} 

In this section, we consider the partial connected dominating set ({\sc Pcds}) 
problem and give a $4 \ln \Delta +2 + o(1)$-approximation algorithm for the same. 
\label{sec:partial}
%{\bf Algorithm.}
\subsection{Algorithm}
We now give a high level overview of the algorithm. 
The algorithm itself is very simple but to show that it is 
indeed a $O(\log \Delta)$ approximation requires non-trivial analysis. 
The algorithm proceeds in the following manner. 
We first run a simple greedy algorithm to find a 
(not necessarily connected) dominating set $D$. 
In each iteration, the greedy algorithm chooses a vertex that dominates the 
maximum number of previously undominated vertices. 
We call this number the  \lq\lq{}profit\rq\rq{} associated with the chosen vertex.  Given this profit function on
the nodes, we now apply a 2-approximation algorithm for the Quota 
Steiner Tree ({\sc Qst}) problem, with  quota of $n'$ to obtain a connected 
solution\footnote{Note that we could have simply defined each node's profit
as the number of vertices it can dominate and then try to connect nodes
using the algorithm for the {\sc Qst}  problem, however in this setting 
there could be a set of high profit nodes that get chosen, but since they all
dominate the {\em same} subset of nodes, we do not actually gain a profit of $n'$.}.

This is a little surprising, since the profit function depends on the
choices made by the greedy algorithm in the first phase. However,
we can show that there is a subset of vertices $D'\subseteq D$, 
of cardinality at most $|\text{\sc Opt}|\ln \Delta + 1$ whose profits sum up to at least 
$n'$ where $|\text{\sc Opt}|$ is the size of the optimum solution of the {\sc Pcds} instance. 
Furthermore the vertices in $D'$ can be connected with additional 
$(\ln \Delta + 1)|\text{\sc Opt}| + 1$ vertices. Thus, if we could find the smallest tree with total profit at least $n'$, such a tree would cost (number of edges in the tree) no more than $\hide{(2\ln \Delta +1)|\text{\sc Opt}| + 2 - 1 = }(2\ln \Delta +1)|\text{\sc Opt}| +1$. This is a special case of the {\sc Qst} problem (with unit edge costs) and hence we can apply Theorem~\ref{thm:QST} to obtain a tree of size (cost) no more than $2((2\ln \Delta +1)|\text{\sc Opt}| + 1) = (4\ln \Delta + 2)|\text{\sc Opt}|+2$. %We note that all the profits are bounded by $n$, hence satisfy the assumption of  theorem~\ref{thm:QST}.
Thus, we obtain a $(4\ln \Delta +2 + o(1))$-approximate solution for the {\sc Pcds} problem.

\vspace{4mm}
\hrule
\vspace{1mm}
\hrule
\begin{algorithm}{ \sc Greedy Profit Labeling Algorithm for {\sc Pcds}.}
\label{algm:PCDS}

\noindent {\bf Input:} Graph $G=(V,E)$ and $n'\in \mathbb{Z}^+\cup\{0\}$.

\noindent {\bf Output:} Tree $T$ with at least $n'$ Coverage.
\vspace{1mm}

\hrule
\vspace{1mm}

\begin{algorithmic}[1]
\State Compute the greedy dominating set $D$ and the corresponding profit function $p:V\rightarrow \mathbb{N}$ using the Algorithm~\ref{algm:greedy}.
\State  Use the 2-approximation for {\sc Qst} problem \cite{johnson2000prize} on the instance $(G,p)$ with quota $n'$ to obtain a tree $T$ % with profit at least $n'$. 
\end{algorithmic}
\vspace{1mm}
\hrule
\vspace{1mm}
\hrule
\end{algorithm}

\vspace{4mm}
\hrule
\vspace{1mm}
\hrule
\begin{algorithm} { \sc Greedy Dominating Set.}
\label{algm:greedy}

\noindent {\bf Input:} Graph $G=(V,E)$.

\noindent {\bf Output:} Dominating Set $D$ and profit function $p:V\rightarrow \mathbb{Z}^+\cup\{0\}$.
\vspace{1mm}
\hrule
\vspace{1mm}
\begin{algorithmic}[1]
\State   $D \leftarrow \phi$ 
 \State  $U \leftarrow V$ 
 \ForAll {$v \in V$}
  \State $p(v) \leftarrow 0$;
   \EndFor
 \While{$U \neq \phi$}
   \State $v \leftarrow \displaystyle \argmax_{v \in V \setminus D} \quad |N_{U}(v)|$ 
\Comment \emph{{$N_U(v)$ is the set of neighbors of $v$, including itself, in the set $U$}}
\State   $C_v \leftarrow N_U(v)$
\State  $p(v) \leftarrow |C_v|$

   \State  $U \leftarrow U \setminus N_U(v)$
   \State  $D \leftarrow D \cup \{v\}$ 
   \EndWhile

\end{algorithmic}
\hrule
\vspace{1mm}
\hrule
\end{algorithm}

\subsection{Analysis}
We first introduce some required notation.

{\bf Notation: } For every vertex $v \in D$ that is chosen by the greedy algorithm, let $C_v$ denote the set of \emph{new} vertices that $v$ dominates %according to the greedy algorithm, 
i.e.,  we have $p(v) = |C_v|$. We say that $v$  \lq\lq{}covers\rq\rq{} a vertex $w$ if and only if $w \in C_v$. 
For the sake of analysis, we partition the vertices of the graph $G$ into layers. 
Let $L_1 = $ {\sc Opt} be the vertices in an optimal solution for the {\sc Pcds} instance, 
$L_2$ be the set of vertices that are not in $L_1$ and have at least one neighbor in $L_1$, and 
$R = V \setminus \{L_1 \cup L_2\}$ be the remaining vertices. 
Let $L_3$ be the subset of vertices  of  $R$ that have a neighbor in $L_2$. 
Furthermore let $L_i' = D \cap L_i, 1 \leq i \leq 3$ where $D$ is the 
dominating set chosen by the greedy algorithm. 
Figure~\ref{fig1} clarifies this notation regarding the layers $L_i$.

\begin{figure}[htbp]

\centering
\includegraphics[width=0.5\textwidth]{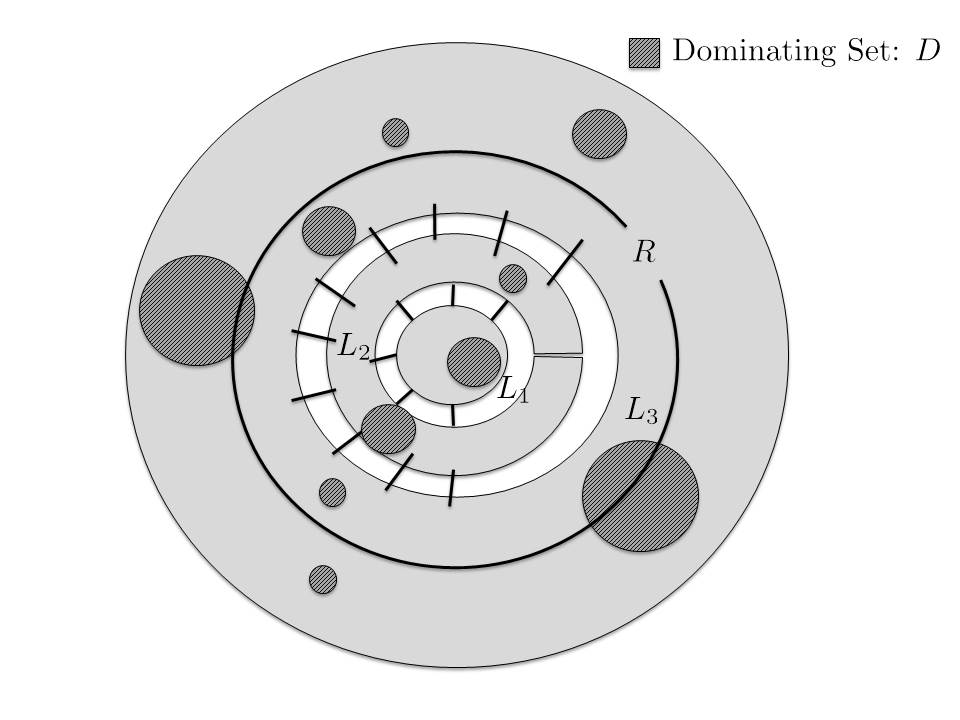}

\caption{Pictorial Representation of Different Layers. 
(a) $L_1$ is an optimal solution  (b) $L_2$ is set of the vertices adjacent to $L_1$ 
(c) $L_3$ is the subsequent layer (d) $R$ is the  set of 
all vertices other than $L_1\cup L_2$ and (e) $L_i' = L_i \cap D$.}
\label{fig1}
\end{figure}

We first show the following. 
\begin{lemma}
There is a subset $D' \subseteq L_1' \cup L_2' \cup L_3'$ such that $|D'| \leq |\text{\sc Opt}|\ln \Delta + 1$ and the total profit of vertices in $D'$ is at least $n'$, i.e. $\sum_{v \in D'} p(v) \geq n'$.
\label{claim:exists}
\end{lemma}
\begin{proof}
Let $L_1' \cup L_2' \cup L_3' = \{v_1, v_2, \ldots, v_l\}$ where the vertices are arranged according to the order in which they were selected by the greedy algorithm. Since all vertices in $L_1 \cup L_2$ are dominated by $L_1' \cup L_2' \cup L_3'$, we have $\sum_{i=1}^l p(v_i) \geq |L_1 \cup L_2| \geq n'$ where the second inequality follows from the fact that $L_1$ is a feasible solution (in fact optimal feasible solution). Choose $t$ such that $\sum_{i=1}^t p(v_i) < n'$ and $\sum_{i=1}^{t+1} p(v_i) \geq n'$. Let $\mathcal{S} = \{v_1, v_2, \ldots, v_t\}$ denote the set of the first $t$ vertices chosen from the set $L_1' \cup L_2' \cup L_3'$. We now show that $|\mathcal{S}| = t \leq |\text{\sc Opt}| \ln \Delta$ and hence $D' = \mathcal{S} \cup \{v_{t+1}\}$ satisfies the requirements of the claim.

Let $C_{12}$ be the set of vertices in $L_1\cup L_2$ that are covered by $\mathcal{S}$ in the original greedy step i.e., $C_{12} = \cup_{v\in \mathcal{S}} \{ C_v\cap(L_1\cup L_2)\}$. Let $UC_{12} = (L_1 \cup L_2) \setminus C_{12}$ be the vertices in $L_1\cup L_2$ that are not covered by $\mathcal{S}$. Similarly define $C_R = \cup_{v \in \mathcal{S}}\{ C_v \cap R\}$ as the set of vertices in $R$ covered by $\mathcal{S}$ (as per the greedy step).
Then, we have that $|C_R| + |C_{12}| < n' \leq |L_1 \cup L_2| = |C_{12}| + |UC_{12}|$, where the first inequality follows from the definition of $\mathcal{S}$. Therefore we have $|C_R| < |UC_{12}|$. 

We can thus assign every vertex in $C_R$ to a unique vertex in $UC_{12}$, i.e. let $I : C_R \rightarrow UC_{12}$ denote a one to one function from $C_R$ to $UC_{12}$. In the subsequent charging argument, any cost that we charge to a vertex $x\in C_R$ is transferred to its assigned vertex $I(x) \in UC_{12}$. Hence, after this charge transfer, only vertices in $L_1 \cup L_2$ will be charged.
We will now use a charging argument to show that $|\mathcal{S}| \leq |\text{\sc Opt}|\ln \Delta$.

%%%%%%%%%%%%%%%%%%%%%%%%%%%CHARGING ARGUEMENT%%%%%%%%%%%%%%%%%%%%%%%%%
Consider a vertex $u\in \mathcal{S}$. We recall that $C_u$ is the set of vertices covered for the first time by $u$ in the greedy step. We assign every $w\in C_u$ a charge $\rho(w) = \frac{1}{|C_u|}$. It is clear that the total charge on all vertices is equal to the size of $\mathcal{S}$. As described above, the charge of a vertex in $w\in R$ is transfered to its mapped vertex in $I(w)\in UC_{12}$. 
Let $v$ be a vertex in the optimal solution set $L_1$. We denote the set of neighbors of $v$, including itself, by $\mathcal{N}(v)$. We claim that the total charge on the vertices of $\mathcal{N}(v)$ is at most $\ln \Delta$. Initially, none of the vertices in $\mathcal{N}(v)$ are charged. Let $u_1, u_2\ldots, u_l$ be the vertices in $\mathcal{S}$ which charge some vertices of $\mathcal{N}(v)$ in that order. This charge could either be the direct charge or a transfer of charge from some vertex in $R$. 
For $i\in[l]$, let $O_i\subseteq \mathcal{N}(v)$ denote the set of vertices that remain uncharged (either directly or through a transfer), after the  vertex $u_i$ is picked into $\mathcal{S}$. 
%We denote the set of uncharged vertices (neither charged directly nor from a transfer) in $\mathcal{N}(v)$ after these vertices are picked, respectively, by $O_1, O_2 \ldots, O_l$. 
Let $O_0 = \mathcal{N}(v)$. 

We will now show that, for every $u_i$, $|C_{u_i}| \geq |O_{i-1}|$. Let us consider the iteration of the greedy algorithm in which $u_i$ is picked.
We claim that none of the vertices in $O_{i-1}$ can be dominated  by any vertex chosen before $u_i$ in the greedy algorithm.  Let $w\in O_{i-1}$ be some vertex which is dominated by some vertex $u'$ chosen by greedy before $u_i$, such that $w\in C_{u'}$. Clearly $u'\in L_1'\cup L_2' \cup L_3'$ should hold, because no vertex in $R\setminus L_3$ can dominate $w$.  But since $u'$ was chosen before $u_i$ and $u'\in L_1'\cup L_2'\cup L_3'$, $u'$ must be chosen into $\mathcal{S}$ before $u_i$. Hence, $w$ cannot be an uncharged vertex in the current iteration leading to a contradiction. %, in this iteration where $u_i$ is chosen and thus $w \notin O_{i-1}$. 
 
Thus, in the iteration where the greedy algorithm was about to choose $u_i$, none of the vertices $O_{i-1}$ have  been dominated. Hence if the greedy were to choose $v$, then $p(v) \geq |O_{i-1}|$. Since the greedy algorithm chooses vertex $u_i$ instead of $v$, we have $|C_{u_i}| \geq |O_{i-1}|$.
% We note that every vertex, picked by greedy algorithm so far, which dominates some vertex of $H(v)$ is also picked by $S$. 
% Therefore the number of vertices in $H_v$ which are not charged directly is equal to the number of new vertices which are covered if $v$ were picked in this iteration. Since the latter is less than $|C_{u_i}|$ (greedy rule), the number of vertices uncharged directly is at most $|C_{u_i}|$. Since the number of vertices uncharged directly is greater than or equal to $O_{i-1}$(because this term involves vertices which were charged indirectly), we have $O_{i-1} \leq |C_{u_i}|$.

The total charge in this iteration ($C_{u_i} \cap \mathcal{N}(v)$) is thus at most $\frac{|O_{i-1}| - |O_{i}|}{|O_{i-1}|}$. Adding these charges over all $l$ iterations, we get, using an analysis very similar to the set cover analysis~\cite{cormengreedy}, $\sum_{w \in \mathcal{N}(v)} \rho(w) \leq H(\Delta)$, where $H$ is the harmonic function and $\Delta$ is the maximum degree.
Adding up the charges over all vertices in $L_1$, we get
 $\sum_{u \in C_{12} \cup UC_{12}} \rho(u) \leq \sum_{v \in L_1} \sum_{w \in \mathcal{N}(v)} \rho(w) \leq |\text{\sc Opt}| \ln \Delta$.  Hence we have $|\mathcal{S}| \leq |\text{\sc Opt}| \ln \Delta$. Since $\mathcal{S}$ was a maximal set having profit at most $n'$, we obtain a set $D'$ with $|D'| = |\mathcal{S}| + 1$ with profit at least $n'$ by adding a single vertex to $\mathcal{S}$, which gives us the desired result. \QEDA
\end{proof}

\begin{theorem}\label{thm:QST2}
Let {\sc Opt} be the optimal solution set for an instance of {\sc Pcds}. There exists a tree $\hat{T}$ with at most $2|\text{\sc Opt}|\ln \Delta +|\text{\sc Opt}|+ 1$ edges such that $\sum_{v\in \hat{T}} p(v) \geq n'$.  
\end{theorem}
\begin{proof}
In Lemma~\ref{claim:exists}, we have shown that there exists a subset $D' \subseteq L_1' \cup L_2' \cup L_3'$ of size $|\text{\sc Opt}|\ln \Delta + 1$ that has profit at least $n'$. However this set $D'$ need not be connected. We now show that this set $D'$ can be connected without paying too much.
 Firstly we note that for every vertex $v\in L_3\cap D'$, there exists a vertex $w\in L_2$ such that $w$ dominates $v$. Thus we can pick a subset $D'' \subseteq L_2$ of size at most $|L_3\cap D'|\leq |\text{\sc Opt}|\ln \Delta + 1$ which dominates all vertices of $L_3\cap D'$.  
Now, it is sufficient to ensure that all the vertices of $(D'\cap L_2)\cup D''$ are connected. This can be achieved by simply adding all the vertices of $L_1$ to our solution. 
Thus we have shown that $\hat{D} = D'\cup D'' \cup L_1$ induces a connected subgraph with profit at least $n'$ and the number of vertices in $\hat{D} \leq |D'| + |D''| + |L_1| \leq 2|\text{\sc Opt}|\ln \Delta + |\text{\sc Opt}| + 2$. Hence there exists subtree $\hat{T}$ on these vertices with at most $(2\ln \Delta +1)|\text{\sc Opt}| + 1$ edges with the requisite total profit. \QEDA
\end{proof}

\begin{corollary}
  Algorithm \ref{algm:PCDS} is a $4\ln \Delta + 2 + o(1)$-approximation algorithm for {\sc Pcds}.
\end{corollary}
\begin{proof}
Let {\sc Opt} be the optimal solution of the {\sc Pcds} instance.  As per Theorem \ref{thm:QST2}, we know that there exists a Steiner tree $\hat{T}$ with at most $2|\text{\sc Opt}|\ln \Delta + |\text{\sc Opt}| + 1$ edges whose total profit exceeds the quota $n'$. Hence, the tree $T$ returned by the 2-approximation for the {\sc Qst} problem has at most $4|\text{\sc Opt}|\ln \Delta + 2|\text{\sc Opt}| +2$ edges. Thus, we obtain a $4\ln \Delta + 2 + o(1)$ approximation algorithm. \QEDA
\end{proof}
% We now use the 2 approximation for the Quota Steiner Tree problem to obtain a Steiner tree of size at most $4k\ln n + 2k +4$ with profit at least $q$. Thus we obtain a $4\ln n + 2$ approximation algorithm.

%\input{input_files/bcds.tex}
\section{Budgeted Connected Dominating Set}
\label{sec:budgeted}
We now turn our attention to the Budgeted Connected Dominating Set ({\sc Bcds}) problem. We recall that in the {\sc Bcds} problem, we have to choose at most $k$ vertices that induce a connected subgraph and maximize the number of dominated vertices.

\subsection{Algorithm}
Algorithm~\ref{algm:BCDS} is very similar to the one we used to obtain a partial connected dominating set. We start by running the standard greedy algorithm to find a dominating set $D$ in the graph. 
We set the profits of vertices in $D$ as the number of newly covered vertices 
at each step of the greedy algorithm, 
while we assign zero profit for the remaining vertices in $V \setminus D$.
In the analysis section, we show that there is a tree on at most $3k$ vertices that has a total profit of at least $(1-\frac{1}{e})\text{\sc Opt}$ where $\text{\sc Opt}$ is the number of vertices dominated by an optimal solution. Note that we may assume that we have guessed $\text{\sc Opt}$ by trying out values between $k$ and $n$ using, say, binary search.
We run the 2 approximation algorithm for the Quota Steiner tree problem on this instance with the quota being set to $(1-\frac{1}{e})\text{\sc Opt}$.  
This will result in a tree with at most $6k$  nodes with total profit at least $(1-\frac{1}{e})\text{\sc Opt}$. Thus we obtain a $(6, 1-\frac{1}{e})$ bicriteria approximation algorithm. 
To convert this bicriteria approximation into a true approximation, we use a dynamic program (Section~\ref{sec:finding-best-subtree}) to find the \lq\lq{}best\rq\rq{} subtree on at most $k$ vertices from this tree of $6k$ vertices. We use a simple tree decomposition scheme to show that the best tree dominates at least $\frac{1}{13}(1- \frac{1}{e})\text{\sc Opt}$ nodes.

\vspace{4mm}
\hrule
\vspace{1mm}
\hrule

\begin{algorithm}{\sc Greedy Profit Labeling Algorithm for {\sc Bcds}.}
\label{algm:BCDS}

\noindent {\bf Input:} Graph $G=(V,E)$ and $k\in \mathbb{N}$.

\noindent {\bf Output:} Tree $\tilde{T}$ with cost at most $k$.

\vspace{1mm}
\hrule
\vspace{1mm}
\begin{algorithmic}[1]
\State Compute the greedy dominating set $D$ and the corresponding profit function $p:V\rightarrow \mathbb{N}$ using the Algorithm~\ref{algm:greedy}.
\State $\text{\sc Opt} \leftarrow $ number of vertices dominated by an optimal solution. \Comment Guess using binary search between $k$ and $n$
\State Use the  2-approximation for {\sc Qst} problem~\cite{johnson2000prize} to obtain a tree $T$ with profit at least $(1 - \frac{1}{e})\text{\sc Opt}$. \Comment 
We show that $|T| \leq 6k$.
\State Use the dynamic program of Section \ref{sec:finding-best-subtree} to find $\tilde{T}$, the best subtree of $T$ having at most $k$ vertices.

\end{algorithmic}
\vspace{1mm}
\hrule
\vspace{1mm}
\hrule
\end{algorithm}

\subsection{Analysis}
Let $L_1$ denote the vertices in an optimal solution. Let layers $L_2$, $L_3$, $R$, and $L_i'$ be defined as in Section~\ref{sec:partial}. $\text{\sc Opt} = |L_1 \cup L_2|$ is the number of vertices dominated by the optimal solution.

Let $L_1' \cup L_2' \cup L_3' = \{v_1, v_2, \ldots, v_l\}$ where the vertices are according to the order in which they were selected by the greedy algorithm. Let $D' = \{v_1, v_2, \ldots, v_k\}$ denote the first $k$ vertices from $L_1' \cup L_2' \cup L_3'$. In Lemma~\ref{lem:budgeted}, we prove that the total profit of $D' = \sum_{v \in D'} p(v)$ is at least $(1 - \frac{1}{e}) \text{\sc Opt}$. Next, we can show that these $k$ vertices can be connected by using at most $2k$ more vertices, thus proving the existence of a tree with at most $3k$ vertices having the desired total profit.

Let $g_i$ denote the total profit after picking the first $i$ vertices from $D'$, i.e., $g_i = \sum_{j=1}^i p(v_j)$. We start by proving that the following recurrence holds for every $i = 0\text{ to }k-1$. 

\begin{claim}
\label{claim:maxcov}
 $g_{i+1} - g_{i} \geq \frac{1}{k} (\text{\sc Opt} - g_i)$   
\end{claim}
\begin{proof}
Consider the iteration of the greedy algorithm, where vertex $v_{i+1}$ is being picked. We first show that at most $g_i$ vertices of $L_1 \cup L_2$ have been already been dominated. Note that any vertex $w \in L_1 \cup L_2$ that has been already dominated must have been dominated by a vertex in $\{v_1, v_2, \ldots v_i\}$. This is because no vertex from $R \setminus L_3$ can neighbor $w$. Since $g_i = \sum_{j=1}^i p(v_j)$ is the total profit gained so far, it follows that at most $g_i$ vertices from $L_1 \cup L_2$ have been dominated. Hence we have that there are at least $\text{\sc Opt} - g_i$ undominated vertices in $L_1 \cup L_2$. Since the $k$ vertices of $L_1$ together dominate all of these, it follows that there exists at least one vertex $v \in L_1$ which neighbors at least $\frac{1}{k}(\text{\sc Opt} - g_i)$ undominated vertices. 

We conclude this proof by noting that since the greedy algorithm chose to pick $v_{i+1}$ at this stage, instead of the $v$ above, it follows that $p(v_{i+1}) = g_{i+1} - g_i \geq \frac{1}{k}(\text{\sc Opt} - g_i)$. \QEDA

\end{proof}

\begin{lemma} \label{lem:budgeted}
%samirjul6 - please define k-tree or elaborate on this point
Let $\text{\sc Opt}$ be the number of vertices dominated by an optimal solution for {\sc Bcds}. Then there exists a subset $D' \subseteq D$ of size $k$ with total profit at least $(1 - \frac{1}{e}) \text{\sc Opt}$. Further, $D'$ can be connected using at most $2k$ 
Steiner vertices. 
\end{lemma}
\begin{proof}
From the Claim~\ref{claim:maxcov}, the profit after $i+1$ iterations is given by 

\begin{equation*}
g_{i+1} \geq \frac{\text{\sc Opt}}{k} + g_i(1 - \frac{1}{k}).
\end{equation*} 
By solving this recurrence, we get  
$g_{i} \geq (1 - (1 - \frac{1}{k})^i) \text{\sc Opt}$. Hence, we obtain the following.
\[\sum_{v \in D'} p(v) = g_k \geq (1 - (1 - \frac{1}{k})^k ) \text{\sc Opt} \geq (1 - \frac{1}{e}) \text{\sc Opt}\]

We show that $D'$ can be connected by at most $2k$ Steiner 
nodes to form a connected tree.
Note that for every vertex $v\in L_3\cap D'$, there exists a vertex $w\in L_2$ such that $w$ neighbors $v$. Thus we can pick a subset $D'' \subseteq L_2$ of size at most $|L_3\cap D'|\leq k$ which dominates all vertices of $L_3\cap D'$. Now, it is sufficient to ensure that all the vertices of $(D'\cap L_2)\cup D''$ are connected. This can be achieved by simply adding all the $k$ vertices of $L_1$. Thus we have shown that $\hat{D} = D'\cup D'' \cup L_1$ induces a connected subgraph with profit at least $(1 - \frac{1}{e}) \text{\sc Opt}$ and $|\hat{D}| \leq |D'| + |D''| + |L_1| \leq 3k$. \QEDA
\end{proof}
\begin{lemma}
\label{lemma:biapprox}
There is a $(6, (1-\frac{1}{e}))$ bicriteria approximation algorithm for the {\sc Bcds} problem. 
\end{lemma}
\begin{proof}
Lemma~\ref{lem:budgeted} shows that there exists a Steiner tree with at most $3k$ vertices having total profit greater than a quota of $(1-\frac{1}{e})\text{\sc Opt}$. Hence, using the $2$-approximation for the {\sc Qst} problem, we obtain a tree $T$ of at most $6k$ nodes and total profit at least $(1 - \frac{1}{e})\text{\sc Opt}$. Thus we obtain a (6, $(1 - \frac{1}{e})$) bicriteria approximation algorithm for the {\sc Bcds} problem. \QEDA
\end{proof}
\subsubsection{Converting the Bicriteria Approximation to a True Approximation}
In order to obtain a true approximate solution (solution of size $k$), we need a technique to find a small subtree $\tilde{T} \subseteq T$ of $k$ vertices which has high total profit. 
% We consider the following general problem - Given a tree $T = (V,E)$ of $n$ vertices, profits on vertices $p : V \rightarrow \mathbb{Z}^+ \cup \{0\}$, and an integer $k$, find a subtree $\tilde{T}$ of $k$ vertices which maximizes the total profit $\tilde{P} = \sum_{v \in \tilde{T}} p(v)$. 
In Section~\ref{sec:finding-best-subtree}, we show that this problem can be easily solved in polynomial time using dynamic programming. However, simply finding the subtree which maximizes the profit is not enough to give a good approximation ratio. We need a way to compare the total profit of the subtree $\tilde{T}$ with the entire profit $P = \sum_{v \in T} p(v)$. We now show that if $n = 6k$, we can obtain a subtree having profit at least $\frac{1}{13} P$.% using a tree separator lemma.
%
%In order to convert the above bi-criteria approximate solution, we need a technique to obtain a small subtree (of $k$ vertices) with large total profit from the tree

The following lemma is well known in folklore and can be easily proven by induction.  It can also be seen as an easy consequence of a theorem by Jordan~\cite{jordan}.
\begin{lemma}[Jordan~\cite{jordan}]
\label{lemma:folklore} 
Given any tree on $n$ vertices, we can decompose it into two trees (by replicating a single vertex) such that the smaller tree has at most $\left \lceil \frac{n}{2} \right \rceil $ nodes and the larger tree has at most $\left \lceil \frac{2n}{3} \right \rceil$ nodes.
\end{lemma}

We now show the following - 
\begin{lemma}
\label{lemma:decomp}
Let $k$ be greater than a sufficiently large constant. Given a tree $T$ with $6k$ nodes, we can decompose it into $13$ trees of size at most $k$ nodes each.
\end{lemma}
\begin{proof}
We use Lemma~\ref{lemma:folklore} to decompose the tree into two trees $T_1$ and $T_2$ such that $|T_1| \leq |T_2|$. In this decomposition, at most one vertex is duplicated, therefore $|T_1| + |T_2| \leq 6k+1$. Also, we have $|T_1| \leq  3k$. We now have two cases:

\noindent {\bf Case 1: $|T_1| \geq 3k-1$}. In this case, $|T_2| \leq 6k+1 -|T_1| \leq 3k+2$. Now repeatedly using Lemma~\ref{lemma:folklore} we can see that $T_1$ can be decomposed into at most $6$ trees and $T_2$ can be decomposed into at most $7$ trees of size at most $k$. This is  shown in the Figure~\ref{fig:case1}. Hence, in this case, we can decompose the tree $T$ into $13$ trees.

\noindent {\bf Case 2: $|T_1| \leq 3k-2$}. In this case, $|T_2| \leq 4k$. In this case, we can decompose $T_1$ into $5$ trees and $T_2$ can be decomposed into $8$ trees. This is shown in Figure~\ref{fig:case2}.  Thus in this case, we can decompose $T$ into $13$ trees. \QEDA
%Starting from the original tree, we recursively decompose any sub-trees which are of size greater than $k$. The figure~\ref{fig2}  in the appendix shows, we obtain $13$ trees in the worst case. 
\end{proof}
\begin{figure}[ptbh]
\centering
\begin{subfigure}{0.4\textwidth}
%\subcaption{Figure 2}
  \includegraphics[width=\textwidth]{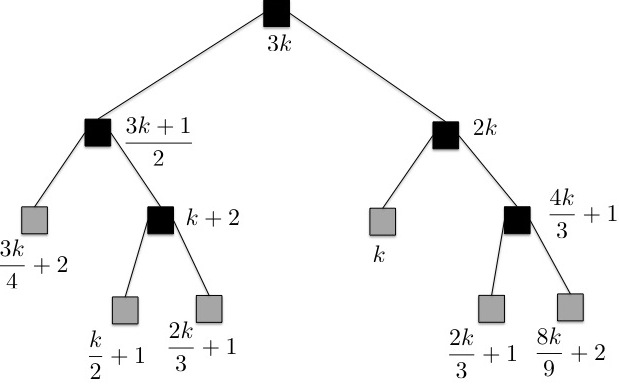}
\end{subfigure}
\quad \quad \quad
\begin{subfigure}{0.4\textwidth}
%\subcaption{Figure 3}
    \includegraphics[width=\textwidth]{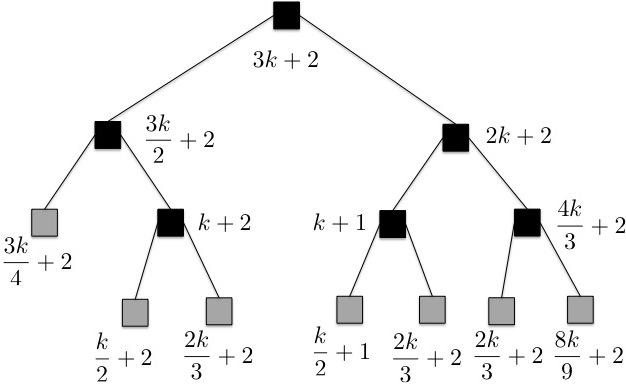}
\end{subfigure}
\caption{ Case 1: $|T_1| \geq 3k-1$. First tree decomposes into $6$ subtrees and second tree decomposes into $7$ trees. In total, we obtain $13$ subtrees. The number associated with each node is the upper bound on the size of the subtree.}
\label{fig:case1}
\end{figure}

\begin{figure}[ptbh]

\centering

\begin{subfigure}[t]{0.4\textwidth}
    \includegraphics[width=\textwidth]{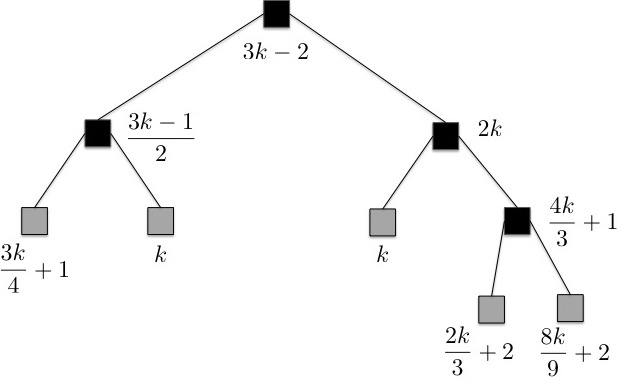}
\end{subfigure}
\quad \quad \quad
%\caption{Figure 1}
%\vspace{-6mm}
\begin{subfigure}[t]{0.4\textwidth}
%\subcaption{Figure 4}
    \includegraphics[width=\textwidth]{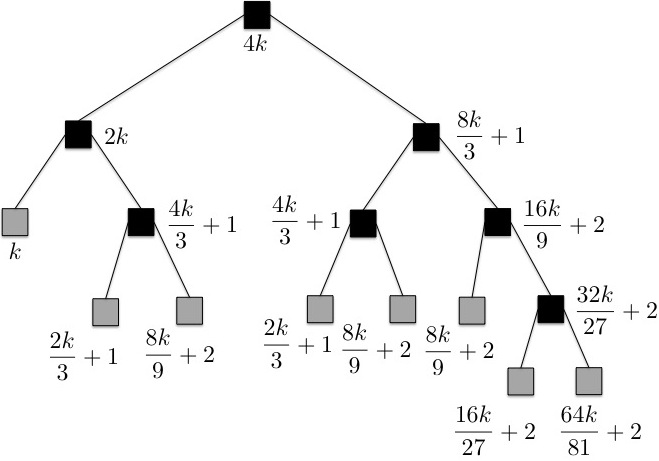}
\end{subfigure}

\caption{Case 2: $|T_1| \leq 3k-2$. First tree decomposes into $5$ subtrees and second tree decomposes into $8$ trees. In total, we obtain $13$ subtrees. The number associated with each node is the upper bound on the size of the subtree.}
\label{fig:case2}

\end{figure}

Using Lemma~\ref{lemma:decomp}, we can convert the bicriteria approximation for {\sc Bcds} to a true approximation algorithm. In particular, we show the following - 
\begin{theorem}
There is a $\frac{1}{13}(1 - \frac{1}{e})$ approximation algorithm for the {\sc Bcds} problem.
\end{theorem}
\begin{proof}
%Given a tree $T$ with $6k$ nodes, we recursively decompose any sub-trees of size greater than $k$ using Lemma~\ref{lemma:folklore}. 
By Lemma~\ref{lemma:biapprox}, we obtain a tree $T$ with at most $6k$ nodes with profit $(1-\frac{1}{e})\text{\sc Opt}$.
Now using Lemma~\ref{lemma:decomp}, we obtain $13$ trees in the worst case, say $T_1, T_2, \ldots T_{13}$. Finally, out of these 13 trees (each of size at most $k$), we pick the tree $\tilde{T}$ with the highest total profit. Let, $p(T) = \sum_{v \in T}p(v)$ denote the total profit of tree $T$. Then we have, 
\[p(\tilde{T}) \geq \frac{1}{13}\sum_{i=1}^{13} p(T_i) \geq \frac{1}{13} p(T) \geq \frac{1}{13}(1 - \frac{1}{e})\text{\sc Opt}\]
Thus we have a $\frac{1}{13}(1 - \frac{1}{e})$ approximation guarantee. \QEDA
\end{proof}

%\begin{Remark}{\bf Finding the Best Subtree}.

\subsubsection{Finding the Best Subtree}
\label{sec:finding-best-subtree}

Although the decomposition Lemma~\ref{lemma:decomp} is useful to prove a theoretical bound, from a practical perspective it is better to use a dynamic programming approach to find the best $k$ sub-tree. Formally, we have the following problem.  
Given a tree $T = (V,E)$ of $n$ vertices, profits on vertices $p : V \rightarrow \mathbb{Z}^+ \cup \{0\}$, and an integer $k$, find a subtree $\tilde{T}$ of $k$ vertices which maximizes the total profit $\tilde{P} = \sum_{v \in \tilde{T}} p(v)$. We show that this problem can be solved in polynomial time using dynamic programming.
Let the tree $T$ be rooted at an arbitrary vertex and $T_v$ denote the subtree rooted at a vertex $v$. We define the following - 

$F(v,i) \leftarrow$ best solution of at most $i$ vertices completely contained inside $T_v$.

$G(v,i) \leftarrow$ best solution of at most $i$ vertices completely contained inside $T_v$ such that $v$ is a part of the solution.

The desired solution is thus at $F(root, k)$. The base cases (when $v$ is a leaf) are trivial. Let $v_1, v_2, \ldots, v_l$ denote the children of vertex $v$. We now have the following recurrence - 
\begin{align}\nonumber   
F(v,i) &= \max \left\{ \max_{1\leq j \leq l} \{F(v_j, i)\} , G(v,i) \right\}\\ \nonumber
G(v,i) &= p(v) + M(l,i-1) \\ \nonumber
\shortintertext{where $M(j,i')$ denotes the best way to distribute a budget of $i'$ among the first $j$ children of $v$. In other words,}
M(l,i-1) &= \max_{i_1+i_2+\ldots+i_l= i-1}  \left\{ \sum_j {G(v_j,i_j)}\right\} \nonumber
\end{align}
$M(j,i')$ is computed using another dynamic program as follows. Again the base cases when $j = 0$ or $i' = 0$ are trivial. For $1\leq j \leq l$ and $1 \leq i'\leq i-1$, we have the following recurrence - 
\begin{align}\nonumber
  M(j,i') &= \max_{0 \leq i^* \leq i'} \left \{ M(j-1, i^*) + G(v_j, i' - i^*)\right \}  
\end{align}

%\end{Remark}

\section{Budgeted Generalized {\bf \sc Cds}}

\label{sec:bgcd}

In this section, we show that our approach extends to more general budgeted connected domination problems. 
Formally, given a graph $G = (V,E)$, a budget $k$, and a monotone \emph{special submodular} profit function $f : 2^V \rightarrow \mathbb{Z^+} \cup{0}$, find a subset $S \subseteq V$ which maximizes $f(S)$ such that $|S| \leq k$ and induces a connected subgraph of $G$. 
\noindent As mentioned earlier in Section \ref{sec:preliminaries}, this problem captures the budgeted variants of the capacitated and weighted profit connected dominating set problems. 

\subsection{ Algorithm.}
Algorithm~\ref{algm:BGCDS} begins by running the standard greedy algorithm to find a basis of the polymatroid associated with $f$. In other words, we greedily pick a vertex $v$ with the maximum marginal profit $f(D \cup \{v\}) - f(D)$ until all vertices have zero marginal profit. With every selected vertex, we associate the marginal profit gained, and associate zero profit with the other vertices. Finally, we run a quota Steiner tree algorithm using these profits to find the smallest tree that yields a profit of at least $(1 - \frac{1}{e}) \text{\sc Opt}$ where $\text{\sc Opt}$ is the optimal profit (which we guess). In the analysis section, we show that there exists a tree $\hat{T}$ of size at most $3k$ with $f(\hat{T}) \geq (1 - \frac{1}{e}) \text{\sc Opt}$. Hence, the 2-approximation for the quota Steiner tree yields a tree $T$ of size at most $6k$ yielding the desired profit. Finally using the tree decomposition described earlier, we show that we can obtain a tree $\tilde{T}$ of size at most $k$ with $f(\tilde{T}) \geq \frac{1}{13}(1 - \frac{1}{e}) \text{\sc Opt}$.

\subsection{\bf Analysis.}
\label{sec:analysis-BGCDS}
Let the $L_1$ denote the vertices in the optimal solution and $f(L_1) = \text{\sc Opt}$. Let $L_2$ denote the set of vertices which have at least one neighbor in $L_1$, and similarly let $L_3$ denote the set of vertices having a neighbor in $L_2$ (and NOT in $L_1$). Let $R = V \setminus \{L_1 \cup L_2 \cup L_3\}$ denote the rest of the vertices. Let $L_i' = D \cap L_i$ where $D$ is the set of vertices chosen by the greedy algorithm. 

Further, let $D'$ denote the first $k$ vertices picked by the greedy algorithm from $L_1' \cup L_2' \cup L_3'$. To simplify notation, let $D' = \{v_1, v_2, \ldots, v_k\}$ and let $D_i$ denote the the set of vertices already picked by the greedy algorithm when the vertex $v_{i+1}$ is being chosen. Hence we have $v_{i+1} = \argmax_{v \in V \setminus D_i} f(D_i \cup \{v\}) - f(D_i)$ and $p(v_{i+1}) = f(D_i \cup \{v_{i+1}\}) - f(D_i)$. Note that in particular $D_i \subseteq D$ but may not be a subset of $D'$. Also let $D'_i = \cup_{j=1}^i v_j$ denote the first $i$ vertices in $D'$. Let $P(D'_i) = \sum_{v \in D'_i} p(v)$ denote the total profit associated with the set $D_i'$. Finally let $D_i'' = D_i \setminus D_i'$ be the vertices in $D_i \cap R$.

\begin{claim}
  $p(v_{i+1}) = P(D_{i+1}') - P(D_{i}') \geq \frac{1}{k}(\text{\sc Opt} - P(D_i'))$
\label{claim:bgcds}
\end{claim}
\begin{proof}\belowdisplayskip=-11pt
Consider the marginal profit of the set $L_1 \setminus D_i'$. Since $N(D_i'')$ does not intersect with $N(L_1)$, we have,
%\vspace{-5mm}
\begin{align}%\label{eq:1} 
  f_{D_i'}(L_1 \setminus D_i') &= f_{D_i' \cup D_i''}(L_1 \setminus D_i') \nonumber \\ \nonumber
%&= f(D_i'\cup D_i''\cup (L\setminus D_i')) - f(D_i'\cup D_i'') \\ \nonumber
%&= f_{D_i''}(D_i'\cup(L\setminus D_i')) + f(D_i'') - f(D_i'\cup D_i'') \\ \nonumber
%&= f_{D_i''}(D_i'\cup(L\setminus D_i')) + f(D_i'') - (f_{D_i''}(D_i') + f(D_i'')) \\ \nonumber
  &= f_{D_i''}(D_i' \cup (L_1 \setminus D_i')) - f_{D_i''}(D_i') \\ \nonumber
  &\geq f_{D_i''}(L_1) - f_{D_i''}(D_i') \\\nonumber
  &= f(L_1) - f_{D_i''}(D_i') \\
\label{eq:1}  &= \text{\sc Opt} - f_{D_i''}(D_i') %\nonumber
\shortintertext{Let us now consider the term $f_{D_i''}(D_i')$. Adding up over successive marginal profits, }\label{eq:3}
f_{D_i''}(D_i') &= \sum_{j=1}^i f_{D_i'' \cup D_{j-1}'}(v_j)
\leq \sum_{j=1}^i f_{D_{j-1}}(v_j) \\\nonumber  
%($D_{j-1} \subseteq D_i'' \cup D_{j-1}'$) 
&= \sum_{j=1}^i p(v_j) = P(D_i') \nonumber
\shortintertext{From Eq~(\ref{eq:1}) and Eq~(\ref{eq:3}), }
f_{D_i'}(L_1 \setminus D_i') &\geq \text{\sc Opt} - P(D_i') \nonumber
\shortintertext{As $f$ is submodular, we have}
f_{D_i'}(L_1 \setminus D_i') &\leq \sum_{w \in L_1 \setminus D_i'} f_{D_i'}(\{w\}) \nonumber
\shortintertext{Since $|L_1 \setminus D_i'| \leq k$, there exists at least one vertex $w \in L_1 \setminus D_i'$ satisfying}\
f_{D_i'}(\{w\}) &\geq \frac{1}{k}  f_{D_i'}(L_1 \setminus D_i') \geq \frac{1}{k} (\text{\sc Opt} - P(D_i')) \nonumber
\shortintertext{Using $f_{D_i}(\{w\}) = f_{D_i'}(\{w\})$ and the fact that greedy picked $v_{i+1}$ at this stage}
p(v_{i+1}) &= f_{D_i}(\{v_{i+1}\}) \geq f_{D_i}(\{w\}) \nonumber\\ 
&\geq  \frac{1}{k} (\text{\sc Opt} - P(D_i')) \nonumber 
\end{align} \QEDA
\end{proof}

%\vspace{-5mm}
Solving the recurrence of Claim~\ref{claim:bgcds}, we have $P(D') \geq (1 - \frac{1}{e})\text{\sc Opt}$. 
%\begin{lemma}
%\label{lemma:BGCDS}
%$P(D') \geq (1 - \frac{1}{e}) {\sc Opt}$
%\end{lemma}

We thus have a set $D'$ of size $k$ which yields a total profit of at least $(1 - \frac{1}{e}) \text{\sc Opt}$. We now proceed to show that the above set $D'$ can be connected at a relatively low cost. Since every vertex in $D'$ can be connected to $L_1$ using at most one vertex (from $L_2$), we can obtain a connected subset $\hat{T}$ of size at most $3k$ by choosing $D'$, $L_1$ and vertices in $L_2$ as described.
Hence, the 2-approximation for the {\sc Qst} problem will yield a tree $T$ of size at most $6k$ which would give a profit of at least $(1 - \frac{1}{e}) \text{\sc Opt}$. Finally applying the tree decomposition described earlier we obtain a tree $\tilde{T}$ of size $\leq k$ with $f(\tilde{T}) \geq P(\tilde{T}) \geq \frac{1}{13} (1 - \frac{1}{e}) \text{\sc Opt}$. 

%%%
\vspace{6mm}
\hrule
\vspace{0.9mm}
\hrule
\begin{algorithm}{ \sc Greedy Profit Labeling Algorithm for {\sc Bgcds}.}
\label{algm:BGCDS}

\noindent {\bf Input:} Graph $G=(V,E)$, a monotone special submodular function $f:2^V\rightarrow \mathbb{Z}^+\cup\{0\}$ and $k\in \mathbb{Z}^+\cup\{0\}$.

\noindent {\bf Output:} Tree $\tilde{T}$ with at most $k$ vertices.
\vspace{1mm}
\hrule
\vspace{1mm}
\begin{algorithmic}[1]
\State Run the Generalized Greedy Dominating Set Routine (Algorithm~\ref{algm:gen-greedy}) on $(G,f)$ to obtain a subset $D$ and a profit function $p:V\rightarrow \mathbb{N}$.
\State $\text{\sc Opt} \leftarrow $ profit of an optimal solution. (Guess using binary search 0 and $f(V)$).
\State $T \leftarrow$ 2-approximation for QST with quota $(1-\frac{1}{e})\text{\sc Opt}$.
% $\hat{T} \leftarrow$ QST($G,p,(1-\frac{1}{e}){\sc Opt}$) \\ 
\State Use the dynamic program of Section \ref{sec:finding-best-subtree} to find $\tilde{T}$, the best subtree of $T$ having at most $k$ vertices.
%\State Use tree separator theorem \cite{jordan} to decompose $T$ into trees of size at most $k$ each. Pick the tree with highest profit and return it as $\tilde{T}$.
%$\tilde{T} \leftarrow separator(T)$ \\
%Return $T$
\end{algorithmic}
\vspace{1mm}
\hrule
\vspace{1mm}
\hrule
\end{algorithm}

%\vspace{4mm}
\hrule
\vspace{1mm}
\hrule
\begin{algorithm} { \sc Generalized Greedy Dominating Set.}
\label{algm:gen-greedy}

\noindent {\bf Input:} Graph $G=(V,E)$ and a monotone special submodular function $f:2^V\rightarrow \mathbb{Z}^+\cup\{0\}$.

\noindent {\bf Output:} $D\subseteq V$ such that $f(D) = f(V)$ and profit function $p:V\rightarrow \mathbb{Z}^+\cup\{0\}$.
\vspace{1mm}
\hrule
\vspace{1mm}
\begin{algorithmic}[1]
\State  $D \leftarrow \phi$ 
 \While{$f(D) \neq f(V)$}
\State  $v \leftarrow \displaystyle \argmax_{v \in V \setminus D} \quad f(D \cup \{v\}) - f(D)$
 \State  $p(v) \leftarrow f(D \cup \{v\}) - f(D)$
  \State $D \leftarrow D \cup \{v\}$
\EndWhile
 \ForAll {$v \in V \setminus D$}
  \State $p(v) \leftarrow 0$
 \EndFor
\end{algorithmic}
\vspace{1mm}
\hrule
\vspace{1mm}
\hrule
\end{algorithm}

\section{Partial Generalized Connected Domination}
We now consider a partial coverage version of the generalized connected domination presented in Section~\ref{sec:bgcd}. In this problem, the goal is to find the smallest subset of vertices which induce a connected subgraph and have total profit at least $q$ (quota).
Just as for the budgeted case, the algorithm proceeds by finding a spanning subset greedily. Using profits as defined by the greedy algorithm, we then find a {\sc Qst} having total profit at least $q$. In the analysis section, we show that there exists a tree $\hat{T}$ of size at most $2k \ln q + k$ with total profit at least $q$. Hence, the 2-approximation for {\sc Qst} yields a tree $T$ of size at most $4k \ln q + 2k$ leading to a $O(4 \ln q)$ approximation.

\subsection{Analysis}
We reuse notation from Section~\ref{sec:bgcd} regarding the layers $L_i$ and $L_i'$. Let $D'$ denote the first $k \ln q + 1$ vertices picked by the greedy algorithm from $L_1' \cup L_2' \cup L_3'$. We now show that the total profit of vertices in $D'$ is at least $q$.
\\
\\

\begin{claim}\label{claim:pgcds} $P(D') \geq q$
\end{claim}
\begin{proof}
As per Claim ~\ref{claim:bgcds}, we obtain the following recurrence
\begin{align}
P(D_{i+1}') &\geq (1 - (1 - \frac{1}{k})^{i+1}) q \\
\intertext{Substituting $i+1 = k \ln q$, we get, }
P(D_{k \ln q}') &\geq (1 - (1 - \frac{1}{k})^{k \ln q}) q \\
&\geq (1 - \frac{1}{q}) q \geq q - 1
\intertext{Since profit function $f$ is integral, we have}
P(D_{k \ln q + 1}') &\geq q
\end{align} \QEDA
\end{proof}

\begin{theorem} Given that the optimal solution is of size $k$, there exists a tree $\hat{T}$ of size at most $2k \ln q + k + 2$ such that $\sum_{v \in \hat{T}} p(v) \geq q$
\end{theorem}
\begin{proof}
In Claim~\ref{claim:pgcds} above, we have demonstrated the existence of a set of size at most $k \ln q + 1$ with the requisite total profit. We now show that this set can be connected at low cost. As in  Theorem~\ref{thm:QST2}, we can see that by selecting at most $k \ln q + 1$ more vertices from layer $L_2$ and at most $k$ vertices from layer $L_1$, the set $D'$ can be connected to form a tree $\hat{T}$. \QEDA
\end{proof}

Finally using the 2-approximation for {\sc Qst}, we obtain a $O(4 \ln q)$ approximation.

\vspace{2mm}
\section{Conclusion and Future Work}
We consider partial and budgeted versions of the well studied connected dominating set problem. We observe that various algorithms which perform well in the  \emph{complete} version of the connected dominating set have unbounded approximation guarantee in the partial case. Using a surprising \emph{greedy profit labeling} algorithm we obtain the first $O(\log n)$ approximation for the partial connected dominating set problem and a $\frac{1}{13}(1-\frac{1}{e})$ approximation for the budgeted version. We also extend our results to a \emph{special submodular} problem, which includes capacitated and weighted profit versions of the {\sc Pcds} and {\sc Bcds} problems as special cases. Our results are tight up to a constant factor in all the cases. A natural open question is to improve these constants.% in these approximation guarantees. 
\\\\
\noindent{\bf Acknowledgment:}
The first author would like to thank Yossi Azar for useful discussions, held during the Dagstuhl seminar on Scheduling (2013),
about the failure of prior methods for the budgeted {\sc Cds} problem.

\bibliographystyle{plain}
\bibliography{cds}

\end{document}